\documentclass[10pt]{article}

\usepackage{amsmath}
\usepackage{amsfonts}
\usepackage{amssymb}
\usepackage{amscd}
\usepackage{amsthm}

\def\card{\parallel\!\!\!\!=}

\def\l{\left}
\def\r{\right}

\def\Rd{\mathbb{R}^d}
\def\Cinf{\mathcal{C}_\infty}

\newtheorem{lemma}{Lemma}[section]
\newtheorem{corollary}[lemma]{Corollary}
\newtheorem{theorem}[lemma]{Theorem}
\newtheorem{proposition}[lemma]{Proposition}
\newtheorem{definition}[lemma]{Definition}
\newtheorem{remark}[lemma]{Remark}

\newcommand{\eh}{\hfill}
\newlength{\sperrT}

\newlength{\sperrTT}

\numberwithin{equation}{section}

\begin{document}

\title{Estimating the number of negative eigenvalues of a
relativistic Hamiltonian
with regular magnetic field}

\date{\today}

\author{Viorel Iftimie, Marius M\u antoiu and Radu Purice\footnote{Institute
of Mathematics ``Simion Stoilow'' of
the Romanian Academy, P.O.  Box 1-764, Bucharest, RO-70700, Romania,
Email: viftimie@math.math.unibuc.ro, mantoiu@imar.ro, purice@imar.ro}}

\maketitle

\begin{abstract}
We prove the analog of the Cwickel-Lieb-Rosenblum estimation for
the number of negative eigenvalues of a relativistic Hamiltonian
with magnetic field $B\in C^\infty_{\rm{pol}}(\mathbb
R^d)$ and an electric potential $V\in L^1_{\rm{loc}}(\mathbb
R^d)$, $V_-\in L^d(\mathbb R^d)\cap L^{d/2}(\mathbb R^d)$. 
Compared to the nonrelativistic case, this estimation involves both norms of $V_-$ in $L^{d/2}(\mathbb R^d)$ and in $L^{d}(\mathbb R^d)$. A
direct consequence is a Lieb-Thirring inequality for the sum of
powers of the absolute values of the negative eigenvalues.
\end{abstract}

\section{Introduction}
\label{sec:Intro}

For the Schr\"odinger operator $-\Delta+V$ on $L^2(\mathbb{R}^d)$
($d\geq3$), one has the well-known CLR (Cwikel-Lieb-Rosenblum)
estimation for $N(V)$, {\it the number of negative eigenvalues}:
\begin{equation}\label{CLR}
 N(V)\ \leq\ c(d)\int_{\mathbb{R}^d}dx\,\left|V_-(x)\right|^{d/2}.
\end{equation}
$V$ is the multiplication operator with the function $V\in
L^1_{\rm{loc}}(\mathbb{R}^d)$ and $V_-:=(|V|-V)/2\in
L^{d/2}(\mathbb{R}^d)$; the constant $c(d)>0$ only depends on the
dimension $d\geq3$ (see \cite{RS}, Th. XII.12).

There exist at least four different proofs of this inequality. Rosenblum \cite{R} uses
"piece-wise polynomial approximation in Sobolev spaces". Lieb \cite{L} relies on the Feynman-Kac formula.
Cwickel \cite{C} uses ideas from interpolation theory. Finally,
Li and Yau \cite{LY} make a heat kernel analysis.

The inequality (\ref{CLR}) has been extended in \cite{AHS} and \cite{S1} to the case of operators with
magnetic fields $(-i\nabla-A)^2+V$, where the components of the vector potential $A=(A_1,\dots,A_d)$
belong to $L^2_{\rm{loc}}(\mathbb R^d)$. The basic ingredient of the proof is the Feynman-Kac-Ito formula.
Melgaard and Rosenblum \cite{MR} generalizes this result (by a different method) to a class of differential
operators of second order with variable coefficients. The idea for treating the relativistic Hamiltonian (without a magnetic field), by replacing Brownian motion with a L\'{e}vy process, appears in \cite{D} and we follow it in our work giving all the technical details. Some similar results but for a different Hamiltonian and with different techniques have been obtained recently in \cite{FLS}.

Our aim in this paper is to obtain an estimation of the type (\ref{CLR}) for an operator that is a good
candidate for a relativistic Hamiltonian with magnetic field (for scalar particles); it is gauge covariant and obtained through a quantization procedure from the classical candidate. We shall make use of
a "magnetic pseudodifferential calculus" that has been introduced and developed in some previous papers \cite{M},
\cite{MP1}, \cite{KO1}, \cite{KO2}, \cite{MP2}, \cite{MP4}, \cite{IMP}.

Let us denote by $C^\infty_{\rm{pol}}(\mathbb R^d)$ the family of functions $f\in C^{\infty}(\mathbb R^d)$ for
which all the derivatives $\partial^\alpha f$, $\alpha\in\mathbb N^d$ have polynomial growth.

Let $B$ be a magnetic field (a $2$-form) with components $B_{jk}\in C^\infty_{\rm{pol}}(\Rd)$. It is
known that it can be expressed as the differential $B=dA$ of a vector potential (a $1$-form)
$A=(A_1,\dots,A_d)$ with $A_j\in C^\infty_{\rm{pol}}(\mathbb R^d)$, $j=1,\dots,d$;
an example is the transversal gauge:
$$
A_j(x)=-\sum_{k=1}^n\int_0^1ds\;B_{jk}(sx)sx_k.
$$

We denote by 

\begin{equation}
\Gamma^A(x,y):=\int^1_0ds\,A((1-s)x+sy)=\int_{[x,y]}A,\ \ x,y\in\Rd.
\end{equation}
 the circulation of $A$ along the segment $[x,y]$, $x,y\in\Rd$.
If $a$ is a symbol on $\mathbb{R}^d$, one defines by an oscillatory integral
the linear continuous operator
$\mathfrak{Op}^A(a):\mathcal S(\mathbb R^{d})\rightarrow\mathcal S^*(\mathbb R^{d})$ by

\begin{equation}\label{Op}
\left[\mathfrak{Op}^A(a)\right](x):=(2\pi)^{-d}\int_{\Rd}\int_{\Rd} dy\,d\xi\,e^{i(x-y)\cdot\xi}
e^{-i\int_{[x,y]}A}a\left(\frac{x+y}{2},\xi\right)u(y),
\end{equation}
The correspondence $a\mapsto \mathfrak{Op}^A(a)$ is meant to be a quantization and could be regarded as
a functional calculus $\mathfrak{Op}^A(a)=a(Q,\Pi^A)$ for the family of non-commuting operators
$(Q_1,\dots,Q_d;\Pi^A_1,\dots,\Pi^A_d)$, where $Q$ is the position operator, $\Pi^A:=D-A(Q)$ is the magnetic
momentum, with $D:=-i\nabla$.

If $a$ belongs to the Schwartz space $\mathcal S(\mathbb R^{2d})$,
then $\mathfrak{Op}^A(a)$ acts continuously in the spaces $\mathcal S(\Rd)$ and $\mathcal S^*(\Rd)$, respectively.
It enjoys the important physical property of being gauge covariant: if $\varphi\in C^\infty_{\rm{pol}}(\Rd)$
is a real function, $A$ and $A':=A+d\varphi$ define the same magnetic field and one prove easily that
$\mathfrak{Op}^{A'}(a)=e^{i\varphi}\mathfrak{Op}^A(a)e^{-i\varphi}$. The property is not shared by the quantization
$a\mapsto \mathfrak{Op}_A(a):=\mathfrak{Op}(a\circ\nu_A)$, where $\mathfrak{Op}$ is the usual Weyl quantization
and $\nu_A:\mathbb R^d\rightarrow\mathbb R^d$, $\nu_A(x,\xi):=(x,\xi-A(a))$ is an implementation of
"the minimal coupling".

We mention that in the references quoted above, a symbolic calculus is developed for the magnetic pseudodifferential
operators (\ref{Op}). In particular, a symbol composition $(a,b)\mapsto a\sharp^B b$ is defined and studied, verifying
$\mathfrak{Op}^A(a)\mathfrak{Op}^A(b)=\mathfrak{Op}^A(a\sharp^B b)$. It depends only on the magnetic field $B$,
no choice of a gauge being needed. The formalism has a $C^*$-algebraic interpretation in terms of twisted crossed
products, cf. \cite{MP1}, \cite{MP3}, \cite{MPR1} and it has been used in \cite{MPR2} for the spectral theory
of quantum Hamiltonians with anisotropic potentials and magnetic fields.

We shall denote by $H_A$ the unbounded operator in $L^2(\Rd)$ defined on $C_0^\infty(\Rd)$ by
$H_Au:=\mathfrak{Op}^A(h)u$, with $h(x,\xi)\equiv h(\xi):=<\xi>-1=(1+|\xi|^2)^{1/2}-1$. One can express it as
\begin{equation}\label{express}
\left(H_Au\right)(x)=(2\pi)^{-d}\int_{\Rd}\int_{\Rd} dy\,d\xi\,e^{i(x-y)\cdot\xi}h\l(\xi-\Gamma^A(x,y)\r)u(y).
\end{equation}
$H_A$ is a symmetric operator and, as seen below, essentially self-adjoint on $C_0^\infty(\Rd)$.
Also denoting its closure by $H_A$, we will have $H_A\ge 0$.

Ichinose and Tamura \cite{IT1}, \cite{IT2}, using the quantization $a\mapsto \mathfrak(Op)_A(a)$,
study another relativistic Hamiltonian with magnetic field defined by
\begin{equation}\label{expres}
\left(H'_Au\right)(x)=(2\pi)^{-d}\int_{\Rd}\int_{\Rd} dy\,d\xi\,e^{i(x-y)\cdot\xi}
h\l(\xi-A\l(\frac{x+y}{2}\r)\r)u(y),
\end{equation}
for which they prove many interesting properties. Unfortunately, $H'_A$ is not gauge covariant (cf. \cite{IMP}). Many of the properties of $H'_A$ also hold for $H_A$ (by replacing
$A\l(\frac{x+y}{2}\r)$ with $\Gamma^A(x,y)$ in the statements and proofs) and this will be used in the sequel.

Aside the magnetic field $B=dA$, we shall also consider an electric potential $V\in L^1_{\rm{loc}}(\Rd)$,
real function expressed as $V=V_+-V_-$, $V_\pm\ge 0$, such that $V_-\in L^{d+k}(\Rd)\cap L^{d/2+k}(\Rd)$ for some $k\ge 0$. We are
interested in the operator $H(A,V):=H_A+V$; it will be shown that it is well-defined in form sense as a self-adjoint
operator in $L^2(\Rd)$, with essential spectrum included into the positive real axis. Taking advantage of gauge
covariance, we denote by $N(B,V)$
the number of strictly negative eigenvalues of $H(A,V)$ (multiplicity counted); it only depends on the potential $V$
and the magnetic field $B$.

The main result of the article is

\begin{theorem}\label{Main}
Let $B=dA$ be a magnetic field with $B_{jk}\in C^\infty_{\rm{pol}}(\Rd)$,
$A_j\in C^\infty_{\rm{pol}}(\Rd)$ and let $V=V_+-V_-\in L^1_{\rm{loc}^(\Rd)}$ be a real function
with $V_\pm\geq0$ and $V_-\in L^{d}(\Rd)\cap L^{d/2}(\Rd)$. Then there exists a constant $C_d$, only depending on the dimension $d\ge 3$, such that
\begin{equation}\label{main}
N(B,V)\le C_d\left(\int_{\Rd} dx\,V_-(x)^d+\int_{\Rd} dx\,V_-(x)^{d/2}\right).
\end{equation}
\end{theorem}

A standard consequence is the next Lieb-Thirring-type estimation:

\begin{corollary}\label{LT}
We assume that the components of $B$ belong to $C^\infty_{\rm{pol}}(\Rd)$ and that
$V=V_+-V_-\in L^1_{\rm{loc}}(\Rd)$ is a real function
with $V_\pm\geq0$ and $V_-\in L^{d+k}(\Rd)\cap L^{d/2+k}(\Rd)$, $k>0$. We denote by $\lambda_1\le\lambda_2\le\dots$ the strictly negative
eigenvalues of $H(A,V)$ (with multiplicity). For any $d\ge 2$ there exists a constant $C_d(k)$ such that
\begin{equation}\label{lt}
\sum_j|\lambda_j|^k\le C_d(k)\left(\int_{\Rd} dx\,V_-(x)^{d+k}+\int_{\Rd} dx\,V_-(x)^{d/2+k}\right).
\end{equation}
\end{corollary}

Sections 2,3,4 will contain essentially known facts (usually presented without proofs), needed for checking
Theorem \ref{Main}. So, in Section $2$ we introduce the Feller semigroup (\cite{IT2}, \cite{Ic2}, \cite{J})
associated to the operator $H_0:=<D>-1$. In the third section we define properly the operator $H(A,V)$
and study its basic properties. In Section 4 we recall
some probabilistic results, as the Markov process associated to the semigroup defined by $H_0$ (\cite{IW}, \cite{DvC}, \cite{J}) and the
Feynman-Kac-It\^o formula adapted to a L\'evy process (\cite{IT2}).

In Section $5$ we prove Theorem \ref{Main}  for $B=0$, using
some of Lieb's ideas for the non-relativistic case (see \cite{S1}) in the setting proposed in \cite{D}. The last section contains the proof of Theorem \ref{Main}
with magnetic field as well as Corollary \ref{LT}. The main ingredient is the Feynman-Kac-It\^o formula.

\section{The Feller semigroup.}
\label{sec:Feller-sg}

We consider the following symbol (interpreted as a classical relativistic Hamiltonian for $m=1, c=1$)
$h:\mathbb{R}^d\rightarrow\mathbb{R}_+$ defined by $h(\xi):=<\xi>-1\equiv\sqrt{1+|\xi|^2}-1$.
Ley us observe (as in \cite{Ic2}) that it defines a {\it conditional negative definite function} (see \cite{RS}) and thus has a L\'{e}vy-Khincin decomposition (see Appendix 2 to Section XIII of \cite{RS}). Computing $(\nabla h)(\xi)$ and $(\Delta h)(\xi)$ and using the general L\'{e}vy-Khincin decomposition (see for example \cite{RS}), one obtains that there exists a L\'{e}vy
measure $\mathsf{n}(dy)$, i.e. a non-negative, $\sigma$-finite measure on
$\mathbb{R}^d$, for which $\min\{1,|y|^2\}$ is
integrable on $\mathbb{R}^d$, such that
\begin{equation}\label{Levy-mes}
 h(\xi)\ =\ -\int_{\mathbb{R}^d} \mathsf{n}(dy)\left\{e^{iy\cdot\xi}-1-i\,(y\cdot\xi) \,I_{\{|x|<1\}}(y)\right\},
\end{equation}
where $I_{\{|x|<1\}}$ is the characteristic function of the open unit
ball in $\mathbb{R}^d$. One has the following explicit
formula (see \cite{Ic2}):
\begin{equation}
 \mathsf{n}(dy)\ =\ 2(2\pi)^{-(d+1)/2}|y|^{-(d+1)/2}K_{(d+1)/2}(|y|)\,dy,
\end{equation}
with $K_{\nu}$ the modified Bessel function of third type and
order $\nu$. We recall the following asymtotic behaviour of these functions:
\begin{equation}\label{iata}
 0\,<\,K_\nu(r)\,\leq\,C\max(r^{-\nu}, r^{-1/2})e^{-r},\quad\forall r>0,\quad\forall\nu>0.
\end{equation}

We shall denote by $\mathcal{H}^s(\mathbb{R}^d)$ the usual Sobolev spaces of order
$s\in\mathbb{R}$ on $\mathbb{R}^d$ and by $H_0$ the pseudodifferential operator $h(D)\equiv\mathfrak{Op}(h)$
considered either as a continuous operator on $\mathcal{S}(\mathbb{R}^d)$ and on
$\mathcal{S}^*(\mathbb{R}^d)$ or as a self-adjoint operator in $L^2(\mathbb{R}^d)$ with domain
$\mathcal{H}^1(\mathbb{R}^d)$. The semigroup generated by $H_0$ is explicitly given by the convolution with the
following function (for $t>0$ and $x\in\mathbb{R}^d$):
$$
 \overset{\circ}{\wp}_t(x)\ :=\ (2\pi)^{-d}\frac{t}{\sqrt{|x|^2+t^2}}\int_{\mathbb{R}^d}d\xi\,
 e^{\l(t-\sqrt{(|x|^2+t^2)(|\xi|^2+1)}\r)}\ =
$$
\begin{equation}\label{iato}
=\ 2^{-(d-1)/2}\,\pi^{-(d+1)/2}\,te^t(|x|^2+t^2)^{-(d+1)/4}K_{(d+1)/2}(\sqrt{|x|^2+t^2})
\end{equation}
(see \cite{IT2}, \cite{CMS}). We have
\begin{equation}
 \overset{\circ}{\wp}_t(x)\ >\ 0\quad\text{and}\quad
 \int_{\Rd}dx\,\overset{\circ}{\wp}_t(x)\ =\ 1.
\end{equation}
From (\ref{iata}) one easily can deduce the following estimation
\begin{equation}\label{vine}
 \exists C>0\quad\text{such that}\quad\overset{\circ}{\wp}_t(0)\ \leq Ct^{-d}(1+t^{d/2}),\quad\forall t>0.
\end{equation}

Let us set
\begin{equation}
 C_\infty(\Rd)\ :=\ \left\{\,f\in C(\Rd)\,\mid\,\lim_{|x|\rightarrow\infty}f(x)=0\,\right\}
\end{equation}
and endow it with the Banach norm $\|f\|_\infty:=\sup_{x\in\Rd}|f(x)|$.
Using the above properties of the function $\overset{\circ}{\wp}_t$ we can extend $e^{-tH_0}$
to a well-defined bounded operator $P(t)$ acting in $C_\infty(\Rd)$.
\begin{remark}\label{Feller}
  One can easily verify that $\{P(t)\}_{t\geq0}$ is a Feller semigroup, i.e.:
\begin{enumerate}
 \item $P(t)$ is a contraction: $\|P(t)f\|_\infty\leq\|f\|_\infty$, $\forall f\in C_\infty(\Rd)$;
 \item $\{P(t)\}_{t\geq0}$ is a semigroup: $P(t+s)=P(t)P(s)$;
\item $P(t)$ preserves positivity: $P(t)f\geq0$ for any $f\geq0$ in $C_\infty(\Rd)$;
\item We have $\lim_{t\searrow0}\|P(t)f-f\|_\infty=0,\ \forall f\in C_\infty(\Rd)$.
\end{enumerate}
\end{remark}

\section{The perturbed Hamiltonian.}
\label{sec:pertHam}

Suppose given a magnetic field of class
$\mathcal{C}^\infty_\mathsf{pol}(\Rd)$ and let us choose a
potential vector $A$, such that $B=dA$,  with components also of
class $\mathcal{C}^\infty_\mathsf{pol}(\Rd)$ (this is always possible, as said before). We shall denote by $H_A$ the operator
$\mathfrak{Op}^A(h)$, considered either as a continuous operator on
$\mathcal{S}(\Rd)$ and on $\mathcal{S}^*(\Rd)$ (by duality)
or as an unbounded operator on $L^2(\Rd)$ with domain
$\mathcal{C}^\infty_0(\Rd)$.

Using the Fourier transform one easily proves that for $u\in\mathcal{C}^\infty_0(\Rd)$:
\begin{equation}
 [H_0u](x)\ =\ -\int_{\Rd}n(dy)\,\left[u(x+y)-u(x)-I_{\{|z|<1\}}(y)\left(y\cdot\partial_xu\right)(x)\right].
\end{equation}
Recalling the definition of $\mathfrak{Op}^A(h)$, we remark that
\begin{equation}
 [H_Au](x)\ =\ \left[\mathfrak{Op}^A(h)u\right](x)\ =\ \left[\mathfrak{Op}(h)
 \left(e^{i(x-.)\cdot\Gamma^A(x,.)}u\right)\right](x)\ =
\end{equation}
$$
=\ \left[H_0\left(e^{i(x-.)\cdot\Gamma^A(x,.)}u\right)\right](x).
$$
Combining the above two equations one gets easily

\begin{equation}
 [H_Au](x)\ =\ -\int_{\Rd}n(dy)\,\left[e^{-iy\cdot\Gamma^A(x,x+y)}u(x+y)-u(x)-\right.
\end{equation}
$$
\left.-I_{\{|z|<1\}}(y)\left(y\cdot(\partial_x-iA(x))u\right)(x)\right].
$$
Repeating the arguments in \cite{Ic2} with $\Gamma^A(x,x+y)$ replacing $A((x+y)/2)$ one proves the following
results similar to those in \cite{Ic2}.

\begin{proposition}
 Considered as unbounded operator in $L^2(\Rd)$, $H_A$ is essential self-adjoint on $\mathcal{C}^\infty_0(\Rd)$.
 Its closure, also denoted by $H_A$, is a positive operator.
\end{proposition}

\begin{proposition}
For any $u\in L^2(\Rd)$ such that $H_Au\in L^1_{\mathsf{loc}}(\Rd)$
$$
 \Re\left[(\mathrm{sign}u)(H_Au)\right]\ \geq\ H_0|u|.
$$
\end{proposition}

Using the method in \cite{S2} we can prove the following result.
\begin{proposition}\label{diamagnetic}
 For any $u\in L^2(\Rd)$ we have:
\begin{enumerate}
 \item for any $\lambda>0$ and for any $r>0$
\begin{equation}\label{unica}
\left|\left(H_A+\lambda\right)^{-r}u\right|\ \leq\ \left(H_0+\lambda\right)^{-r}|u|;
\end{equation}
\item for any $t\geq0$
\begin{equation}\label{doica}
\left|e^{-tH_A}u\right|\ \leq\ e^{-tH_0}|u|.
\end{equation}
\end{enumerate}
\end{proposition}

We associate to $H_A$ its sesquilinear form

\begin{equation*}
\mathcal{D}(\mathfrak{h}_A)=\mathcal{D}(H_A^{1/2}),
\end{equation*}

\begin{equation}
\mathfrak{h}_A(u,v):=(H_A^{1/2}u,H_A^{1/2}v),\quad\forall(u,v)\in\mathcal{D}(\mathfrak{h}_A)^2.
\end{equation}

Consider now a function $V\in L^1_{\mathsf{loc}}(\Rd)$, $V\geq0$ and associate to it the sesquilinear form

\begin{equation*}
 \mathcal{D}(\mathfrak{q}_V):=\{u\in L^2(\Rd)\,\mid\,\sqrt{V}u\in L^2(\Rd)\},
\end{equation*}

\begin{equation}\label{qformV}
 \mathfrak{q}_V(u,v):=\int_{\Rd}dx\,V(x)u(x)\overline{v(x)},\ \
\forall(u,v)\in\mathcal{D}(\mathfrak{q}_V)^2.
\end{equation}

Both these sesquilinear forms are symmetric, closed and positive.
We shall abbreviate $\mathfrak{h}_A(u)\equiv\mathfrak{h}_A(u,u)$ and
$\mathfrak{q}_V(u)\equiv\mathfrak{q}_V(u,u)$.

\begin{proposition}\label{Ham-rel-pert}
 Let $V:\Rd\rightarrow\mathbb{R}$ be a measurable function that can be decomposed as $V=V_+-V_-$ with
 $V_\pm\geq0$ and $V_\pm\in L^1_{\mathsf{loc}}(\Rd)$. Moreover let us suppose that the sesquilinear form
 $\mathfrak{q}_{V_-}$ is small with respect to $\mathfrak{h}_0$ (i.e. it is $\mathfrak{h}_0$-relatively
 bounded with bound strictly less then 1). Then the sesquilinear form
 $\mathfrak{h}_A+\mathfrak{q}_{V_+}-\mathfrak{q}_{V_-}$, that is well defined on
 $\mathcal{D}(\mathfrak{h}_A)\bigcap\mathcal{D}(\mathfrak{q}_{V_+})$, is symmetric, closed and bounded from below, defining thus an inferior semibounded self-adjoint operator $H(A;V)\equiv H:=H_A\dotplus V$ (sum in sense
 of forms).
\end{proposition}
\begin{proof}
 The sesquilinear form $\mathfrak{h}_A+\mathfrak{q}_{V_+}$ (defined on the intersection of the form domains) is
 clearly positive, symmetric and closed. We shall prove now that the sesquilinear form $\mathfrak{q}_{V_-}$ is
 $\mathfrak{h}_A+\mathfrak{q}_{V_+}$-bounded with bound strictly less then 1, so that the conclusion of the
 proposition follows by standard arguments.

 Let us denote by $H_+:=H_A\dotplus V_+$ the unique positive
 self-adjoint operator associated to the sesquilinear form $\mathfrak{h}_A+\mathfrak{q}_{V_+}$ by the
 representation theorem 2.6 in \S VI.2 of \cite{K}. As $V_+\in L^1_{\mathsf{loc}}(\Rd)$, we have
 $\mathcal{C}^\infty_0(\Rd)\subset\mathcal{D}(\mathfrak{h}_A)\bigcap\mathcal{D}(\mathfrak{q}_{V_+})$
 and thus we can use the form version of the Kato-Trotter formula from \cite{KM}:

\begin{equation}\label{Kato-Trotter}
 e^{-tH_+}\ =\ \underset{n\rightarrow\infty}{s-\lim}\left(e^{-(t/n)H_A}\,e^{-(t/n)V_+}\right)^n,\qquad
 \forall t\geq0.
\end{equation}

Let us recall the formula ($r>0$ and $\lambda>0$)

\begin{equation}\label{rez-sgroup}
 (H_++\lambda)^{-r}\ =\ \Gamma(r)^{-1}\int_0^\infty dt\ t^{r-1}\,e^{-t\lambda}\,e^{-tH_+}.
\end{equation}

Combining the above two equalities we obtain

\begin{equation}\label{diamagnetic-pert}
 \left|(H_++\lambda)^{-r}f\right|\ \leq\ \Gamma(r)^{-1}\int_0^\infty dt\ t^{r-1}\,e^{-t\lambda}\,
 \left|e^{-tH_+}f\right|\ =
\end{equation}

$$
=\ \Gamma(r)^{-1}\int_0^\infty dt\ t^{r-1}\,\left|\underset{n\rightarrow\infty}{s-\lim}\left(e^{-(t/n)H_A}
\,e^{-(t/n)V_+}\right)^nf\right|\ \leq
$$
$$
\leq\ (H_0+\lambda)^{-r}|f|,
$$
by using the second point of Proposition \ref{diamagnetic}.

Taking $u=(H_0+\lambda)^{-1/2}g$ with $g\in L^2(\Rd)$ arbitrary and $\lambda>0$ large enough and using the
hypothesis on $V_-$ we deduce that there exists $a\in[0,1)$, $b\geq0$ and $a^\prime\in[0,1)$ such that
$$
\mathfrak{q}_{V_-}(u)\leq a\|H_0^{1/2}u\|^2+b\|u\|^2=a\|H_0^{1/2}(H_0+\lambda)^{-1/2}g\|^2+b\|(
H_0+\lambda)^{-1/2}g\|^2\leq
$$
\begin{equation}\label{q-relmarg}
 \leq (a+b/\lambda)\|g\|^2\leq a^\prime\|g\|^2.
\end{equation}
For any $v\in\mathcal{D}(\mathfrak{h}_A)\bigcap\mathcal{D}(\mathfrak{q}_{V_+})$ let $f:=(H_++\lambda)^{1/2}v$
and $g:=|f|$. Using now (\ref{diamagnetic-pert}) with $r=1/2$, (\ref{q-relmarg}) and the explicit form of
$\mathfrak{q}_{V_-}$ we conclude that
\begin{equation}
 \mathfrak{q}_{V_-}(v)= \mathfrak{q}_{V_-}\left((H_++\lambda)^{-1/2}f\right)\leq\mathfrak{q}_{V_-}
 \left((H_0+\lambda)^{-1/2}g\right)\leq
\end{equation}
$$
\leq a^\prime\|g\|^2\ =\ a^\prime\left\|(H_++\lambda)^{1/2}v\right\|^2\ =\ a^\prime\left[\mathfrak{h}_A(v)+
\mathfrak{q}_+(v)+\lambda\|v\|^2\right].
$$
\end{proof}

\begin{definition}\label{def-Ham-rel-pert}
 For a potential function $V$ satisfying the hypothesis of Proposition \ref{Ham-rel-pert}, we call the
 operator $H=H(A;V)$ introduced in the same proposition the relativistic Hamiltonian with potential $V$
 and magnetic vector potential $A$.
\end{definition}

The spectral properties of $H$ only depend on the magnetic field $B$,
different choices of a gauge giving unitarly equivalent
Hamiltonians, due to the gauge covariance of our quantization
procedure.

\begin{proposition}\label{HMagnSEss}
 Let $B$ be a magnetic field with $\mathcal{C}^\infty_{\mathsf{pol}}(\Rd)$ components and $A$ a vector
 potential for $B$ also having $\mathcal{C}^\infty_{\mathsf{pol}}(\Rd)$ components.
 Assume that $V:\Rd\rightarrow\mathbb{R}$ is a measurable function that can be decomposed as $V=V_+-V_-$ with
 $V_\pm\geq0$, $V_+\in L^1_{\mathsf{loc}}(\Rd)$ and $V_-\in L^p(\Rd)$ with $p\geq d$. Then
\begin{enumerate}
 \item $\mathfrak{q}_{V_-}$ is a $\mathfrak{h}_0$-bounded sesquilinear form with relative bound 0;
 \item the Hamiltonian $H$ defined in Definition \ref{def-Ham-rel-pert} is bounded from below and we have
 $\sigma_{\mathsf{ess}}(H)=\sigma_{\mathsf{ess}}(H_A\dotplus V_+)\subset[0,\infty)$.
\end{enumerate}
\end{proposition}

\begin{proof}
\noindent{1.} Using Observation 3 in
\S 2.8.1 from \cite{T}, we conclude that for $d>1$, the Sobolev
space $\mathcal{H}^{1/2}(\Rd)$ (that is the domain of the
sesquilinear form $\mathfrak{h}_0$) is continuously embedded in
$L^r(\Rd)$ for $2\leq r\leq 2d/(d-1)<\infty$. Also using
H\"{o}lder inequality, we deduce that for
$r=2p/(p-1)\in[2,2d/(d-1)]$, for $p\geq d$

\begin{equation}\label{3.14}
 \|V_-^{1/2}u\|_2^2\,\leq\,\|V_-\|_p\|u\|_r^2\,\leq\,c\|V_-\|_p\|u\|_{\mathcal{H}^{1/2}(\Rd)}^2,
\end{equation}
$\forall u\in\mathcal{H}^{1/2}(\Rd)=\mathcal{D}(\mathfrak{h}_0)$.
Thus $V_-^{1/2}\in\mathbb{B} (\mathcal{H}^{1/2}(\Rd);L^2(\Rd))$; now let us prove that it is even compact.
Let us observe that for $d\leq p<\infty$,
$\mathcal{C}^\infty_0(\Rd)$ is dense in $L^p(\Rd)$. Thus, for
$d\leq p<\infty$ let
$\{W_\epsilon\}_{\epsilon>0}\subset\mathcal{C}^\infty_0(\Rd)$ be
an approximating family for $V_-^{1/2}$ in $L^{2p}(\Rd)$, i.e.
$\|V_-^{1/2}-W_\epsilon\|_{2p}\leq\epsilon$. Moreover, for any sequence $\{u_j\}\subset\mathcal{H}^{1/2}(\mathbb{R}^d)$ contained in the unit ball (i.e. $\|u_j\|_{\mathcal{H}^{1/2}}\leq1$) we may suppose that it converges to $u\in\mathcal{H}^{1/2}(\mathbb{R}^d)$ for the weak topology on $\mathcal{H}^{1/2}(\mathbb{R}^d)$ and thus $\|u\|_{\mathcal{H}^{1/2}}\leq1$. It follows that $W_\epsilon u_j$ converges to $W_\epsilon u$ in $L^2(\mathbb{R}^d)$ and due to (\ref{3.14}) we have:
$$
\|(V_-^{1/2}-W_\epsilon)(u-u_j)\|\leq C^{1/2}\|V_-^{1/2}-W_\epsilon\|_{L^{2p}}\|u-u_j\|_{\mathcal{H}^{1/2}}\leq2c^{1/2}\epsilon,\quad\forall j\geq1.
$$
We conclude that $V_-^{1/2}u_j$ converges in $L^2(\mathbb{R}^d)$ to $V_-^{1/2}u$ and using the duality we also get that $V_-$ is a compact operator from $\mathcal{H}^{1/2}(\mathbb{R}^d)$ to $\mathcal{H}^{-1/2}(\mathbb{R}^d)$. Using exercise
39 in ch. XIII of \cite{RS} we deduce that $\mathfrak{q}_-$ has zero relative bound with respect to $\mathfrak{h}_0$.

\noindent{2.} The conclusion of point 1 implies that the
operator $V_-^{1/2}(H_0+1)^{-1/2}\in\mathbb{B} [L^2(\Rd)]$ is
compact. Using the first  point of Proposition \ref{diamagnetic}
with $\lambda=-1$ and $r=1/2$, and Pitt Theorem in \cite{P}, we
conclude that the operator $V_-^{1/2}(H_A\dotplus
V_++1)^{-1/2}\in\mathbb{B}[L^2(\Rd)]$ is also compact. Thus
$V_-:\mathcal{D}(\mathfrak{h}_A+\mathfrak{q}_{V_+})\rightarrow\mathcal{D}(\mathfrak{h}_A+\mathfrak{q}_{V_+})$
is compact and the conclusion (2) follows from exercise
39 in ch. XIII of \cite{RS}.
\end{proof}

\section{The Feynman-Kac-It\^{o} formula.}
\label{sec:FKI}

In this section we gather some probabilistic notions and results needed in the proof of Theorem \ref{Main}. The main idea is that we obtain a Feynman-Kac-It\^{o} formula (following \cite{IT2}) for the semigroup defined by $H(A,V)$ and this allows us to reduce the problem to the case $B=0$.
For this last one we repeat then the proof in \cite{D} giving all the necessary details for the case of singular potentials $V$; here an essential point is an explicit formula for the integral kernel of the operator $e^{-tH(0,V)}$ in terms of a L\'{e}vy process.

Let $(\Omega,\mathfrak{F},\mathsf{P})$ be a probability space, i.e. $\mathfrak{F}$ is a $\sigma$-algebra of
subsets of $\Omega$ and $\mathsf{P}$ is a non-negative $\sigma$-aditive function on $\mathfrak{F}$ with
$\mathsf{P}(\Omega)=1$. For any integrable random variable $X:\Omega\rightarrow\mathbb{R}$ we denote its
expectation value by

\begin{equation}
 \mathsf{E}(X)\,:=\,\int_\Omega\ X(\omega)\mathsf{P}(d\omega).
\end{equation}

For any sub-$\sigma$-algebra $\mathfrak{G}\subset\mathfrak{F}$ we denote its associated conditional expectation
by $\mathsf{E}(X\mid\mathfrak{G})$; this is the unique $\mathfrak{G}$-measurable random variable
$Y:\Omega\rightarrow\mathbb{R}$ satisfying

\begin{equation}
\int_B Y(\omega)\mathsf{P}(d\omega)\,=\,\int_B X(\omega)\mathsf{P}(d\omega),\qquad\forall B\in\mathfrak{G}.
\end{equation}
Let us recall the following properties of the conditional
expectation (see for example \cite{J}):

\begin{equation}\label{expect-cond-expect}
 \mathsf{E}\left( \mathsf{E}(X\mid\mathfrak{G})\right)\ =\ \mathsf{E}(X),
\end{equation}

\begin{equation}\label{cond-expect-prod}
 \mathsf{E}(XZ\mid\mathfrak{G})\ =\ Z\mathsf{E}(X\mid\mathfrak{G}),
\end{equation}
for any $\mathfrak{G}$-measurable random variable $Z:\Omega\rightarrow\mathbb{R}$, such that $ZX$ is integrable.

We also recall the Jensen inequality (\cite{S1}, \cite{J}): for
any convex function $\varphi:\mathbb{R}\rightarrow\mathbb{R}$, and for any lower bounded random variable
$X:\Omega\rightarrow\mathbb{R}$ the following inequality is valid
\begin{equation}
 \varphi(\mathsf{E}(X))\ \leq\ \mathsf{E}(\varphi(X)).
\end{equation}

Following \cite{DvC}, we can associate to our Feller semigroup
$\{P(t)\}_{t\geq0}$, defined in Section 2, a Markov process
$\left\{(\Omega,\mathfrak{F},\mathsf{P}_x),\{X_t\}_{t\geq0},\{\theta_t\}_{t\geq0}\right\}$;
that we briefly recall here:
\begin{itemize}
 \item $\Omega$ is the set of {\sl "cadlag"} functions on $[0,\infty)$, i.e. functions
 $\omega:[0,\infty)\rightarrow\Rd$ (paths) that are continuous to the right and have a limit to the left
 in any point of $[0,\infty)$.
 \item $\mathfrak{F}$ is the smallest $\sigma$-algebra for which all the {\sl coordinate functions}
 $\{X_t\}_{t\geq0}$, with $X_t(\omega):=\omega(t)$, are measurable.
 \item $\mathsf{P}_x$ is a probability on $\Omega$ such that for any $n\in\mathbb{N}^*$, for any
 ordered set $\{0< t_1\leq\ldots\leq t_n\}$ and any family $\{B_1,\ldots,B_n\}$ of Borel subsets in $\Rd$,
 we have
\begin{equation}
 \mathsf{P}_x\left\{X_{t_1}\in B_1,\ldots ,X_{t_n}\in B_n\right\}\ =
\end{equation}
$$
\hspace{-1.5cm}
=\ \int_{B_1}dx_1\,\overset{\circ}{\wp}_{_{t_1}}(x-x_1)\int_{B_2}dx_2\,\overset{\circ}{\wp}_{_{t_2-t_1}}
(x_1-x_2)\,\ldots\,\int_{B_n}dx_n\,\overset{\circ}{\wp}_{_{t_n-t_{n-1}}}(x_{n-1}-x_n).
$$
One can deduce that, if $\mathsf{E}_x$ denotes the
expectation value with respect to $\mathsf{P}_x$, then for any
$f\in\Cinf(\Rd)$ and for any $t\geq0$ one has
\begin{equation}
 \mathsf{E}_x(f\circ X_t)\ =\ \left[P(t)f\right](x).
\end{equation}
We also remark that $\mathsf{P}_x$ is the image of the probability
$\mathsf{P}_0\equiv\mathsf{P}$ under the map
$S_x:\Omega\rightarrow\Omega$ defined by
$\left[S_x\omega\right](t):=x+\omega(t)$.
 \item For any $t\geq0$, the map $\theta_t:\Omega\rightarrow\Omega$ is defined by
 $\left[\theta_t\omega\right](s):=\omega(s+t)$. If we denote by $\mathfrak{F}_t$ the sub-$\sigma$-algebra of
 $\mathfrak{F}$ generated by the processes $\{X_s\}_{0\leq s\leq t}$, then for any $t\geq0$ and
 any bounded random variable $Y:\Omega\rightarrow\mathbb{R}$
\begin{equation}\label{Markov}
 \mathsf{E}_x\left(Y\circ\theta_t\mid\mathfrak{F}_t\right)(\omega)\ =\ \mathsf{E}_{X_t(\omega)}(Y),\quad
 \mathsf{P}_x-a.e.\ {\rm on}\ \Omega.
\end{equation}
\end{itemize}

We use the fact that (see \cite{IW}, \cite{IT2}) the probability $\mathsf{P}_x$ is concentrated on the set
of paths $X_t$ such that $X_0=x$ and by the L\'{e}vy-Ito Theorem:
\begin{equation}
 X_t=x+\int_0^{t_+}\int_{\mathbb{R}^d}y\,\tilde{N}_X(ds\,dy).
\end{equation}
Here $\tilde{N}_X(ds\,dy):=N_X(ds\,dy)-\hat{N}_X(ds\,dy)$, $\hat{N}_X(ds\,dy):=\mathsf{E}_x(N_X(ds\,dy))=ds\,\mathsf{n}(dy)$ with 
$\mathsf{n}(dy)$ the L\'{e}vy measure appearing in (\ref{Levy-mes}) and $N_X$ a 'counting measure' on $[0,\infty)\times\Rd$ that for $0<t<t^\prime$ and $B$ a Borel subset of $\Rd$ is defined as $N_X((t,t^\prime]\times B)\ :=$
\begin{equation}
:=\ \card\left\{s\in(t,t^\prime]\,\mid\,X_s\ne X_{s-},\,X_s\- X_{s-}\in B\right\}.
\end{equation}

Following the procedure developped in \cite{IT2} by Ichinose and Tamura one obtains a Feynman-Kac-It\^{o} formula for Hamiltonians of the type $H=H_A\dotplus V$. In fact we have

\begin{proposition}\label{FKI}
 Under the same conditions as in Definition \ref{def-Ham-rel-pert}, for any function $u\in L^2(\Rd)$ we have
\begin{equation}
 \left(e^{-tH}u\right)(x)\ =\ \mathsf{E}_x\left((u\circ X_t)\,e^{-S(t,X)}\right),\quad t\geq0,x\in\mathbb{R}^d
\end{equation}
where
$$
S(t,X)\ :=\ i\int_0^{t_+}\int_{\mathbb{R}^d}\tilde{N}_X(ds\,dy)
\left<\int_0^1dr\,\left(A(X_{s_-}+ry)\right),\,y \right>\,+
$$
$$
+\,i\int_0^{t}\int_{\mathbb{R}^d}\hat{N}_X(ds\,dy)
\left<\left(\int_0^1dr\,A(X_s+ry)-A(X_s)\right),\,y \right>\,+
$$
\begin{equation}\label{def-S}
+\,\int_0^tds\,V(X_s).
\end{equation}
\end{proposition}

In the sequel we shall take $A=0$ and
$V\in C_0^\infty(\Rd)$. As it is proved in \cite{DvC}, the operator $e^{-t(H_0\dotplus V)}$ has an integral
kernel that can be described in the following way. Let us denote by $\mathfrak{F}_{t-}$ the sub-$\sigma$-algebra
of $\mathfrak{F}$ generated by the random variables $\{X_s\}_{0\leq s<t}$. For any pair $(x,y)\in[\Rd]^2$ and
any $t>0$ we define a measure $\mu^{t,y}_{0,x}$ on the Borel space $(\Omega,\mathfrak{F}_{t-})$ by the equality
\begin{equation}\label{Def-mu}
 \mu^{t,y}_{0,x}(M)\ :=\ \mathsf{E}_x\left[\chi_M\,\overset{\circ}{\wp}_{t-s}(X_s-y)\right],
\end{equation}
for any $M\in\mathfrak{F}_{s}$ and $0\leq s<t$, where $\chi_M$ is the characteristic function of $M$. This
measure is concentrated on the family of {\sl 'paths'} $\{\omega\in\Omega\mid X_0(\omega)=x,X_{t-}(\omega)=y\}$
and we have $\mu^{t,y}_{0,x}(\Omega)=\overset{\circ}{\wp}_t(x-y)$.

\begin{proposition}\label{FK}
 Let $F:\Omega\rightarrow\mathbb{R}$ be a non-negative $\mathfrak{F}_{t-}$-measurable random variable and
 let $f:\Rd\rightarrow\mathbb{R}$ be a positive borelian function. Then the following equality holds for any
 $t>0$ and any $x\in\Rd$:
\begin{equation}
 \int_{\Rd}dy\left\{\int_\Omega\mu^{t,y}_{0,x}(d\omega)\,F(\omega)\,e^{-\int_0^t ds\,V(X_s)}\right\}\,f(y)\ =
\end{equation}
$$
=\ \mathsf{E}_x\left(F\,e^{-\int_0^t ds\,V(X_s)}\,f(X_t)\right).
$$
\end{proposition}
\begin{proof}
 This is a direct consequence of relations (2.29) and (2.33) from \cite{DvC}.
\end{proof}

Let us now take $A=0$ in Proposition \ref{FKI} and $F=1$ in Proposition \ref{FK} in order to deduce
that the operator $e^{-t(H_0\dotplus V)}$ is an integral operator with integral kernel given by the function
\begin{equation}
\wp_t(x,y)\,:=\, \int_\Omega\mu^{t,y}_{0,x}(d\omega)\,e^{-\int_0^t ds\,V(X_s)},\ \ \;t>0,\ (x,y)\in\Rd\times\Rd.
\end{equation}
Proposition 3.3 from \cite{DvC} implies that the function
$[0,\infty)\times\Rd\times\Rd\ni(t,x,y)\mapsto\wp_t(x,y)\in\mathbb{R}$ is non-negative, continuous and verifies
$\wp_t(x,y)=\wp_t(y,x)$. We shall also need the following result.

\begin{proposition}\label{4.3}
 For any $t>0$, any $x\in\Rd$ and any function $g:\Omega\rightarrow\mathbb{R}$ that is integrable with respect
 to the measure $\mu^{t,x}_{0,x}$ we have the equality:
\begin{equation}
 \int_\Omega\mu^{t,x}_{0,x}(d\omega)\,g(\omega)\ =\ \int_\Omega\mu^{t,0}_{0,0}(d\omega)\,g(x+\omega).
\end{equation}
\end{proposition}

\begin{proof}
 It is evidently sufficient to prove that for any $s\in[0,t)$ and any $M\in\mathfrak{F}_s$ we have
$$
\mu^{t,x}_{0,x}(M)\ =\ \left(\mu^{t,0}_{0,0}\circ S_x^{-1}\right)(M)
$$
where the map $S_x:\Omega\rightarrow\Omega$ is defined by $(S_x(\omega)(t):=x+\omega(t)$. We noticed previously
the identity $\mathsf{P}_x=\mathsf{P}_0\circ S_x^{-1}$; thus for any function
$F:\Omega\rightarrow\mathbb{R}$ integrable with respect to $\mathsf{P}_x$ we have
$\mathsf{E}_x(F)=\mathsf{E}_0(F\circ S_x)$. We remark that $X_s(\omega+x)=\omega(s)+x=X_s(\omega)+x$, and using
the definition of the measure $\mu^{t,x}_{0,x}$ in (\ref{Def-mu}), we obtain
\begin{equation}
 \mu^{t,x}_{0,x}(M)=\mathsf{E}_x\left[\chi_M\,\overset{\circ}{\wp}_{t-s}(X_s-x)\right]=
 \mathsf{E}_0\left[(\chi_M\circ S_x)\,\overset{\circ}{\wp}_{t-s}(X_s)\right]=
\end{equation}
$$
=\mathsf{E}_0\left[(\chi_{S_x^{-1}(M)}\,\overset{\circ}{\wp}_{t-s}(X_s)\right]=
\mu^{t,0}_{0,0}\left(S_x^{-1}(M)\right)=\left[\mu^{t,0}_{0,0}\circ S_x^{-1}\right](M).
$$
\end{proof}

\section{Proof of the bound for $N(0;V)$.}

In this Section we will consider $A=0$ and we shall work only with
a potential $V=V_+-V_-$ satisfying the properties:
\begin{itemize}
 \item $V_\pm\geq0$,
 \item $V_+\in L^1_{\mathsf{loc}}(\Rd)$,
 \item $V_-\in L^d(\Rd)\cap L^{d/2}(\Rd)$.
\end{itemize}

\noindent We shall use the notations $H:=H_0\dotplus V$, $H_+:=H_0\dotplus V_+$, $H_-:=H_0\dotplus(-V_-)$ for
the operators associated to the sesquilinear forms $\mathfrak{h}=\mathfrak{h}_0+\mathfrak{q}_{V}$,
$\mathfrak{h}_+=\mathfrak{h}_0+\mathfrak{q}_{V_+}$, $\mathfrak{h}_-=\mathfrak{h}_0-\mathfrak{q}_{V_-}$.

Due to the results of Proposition \ref{HMagnSEss} we have $\sigma_{\mathsf{ess}}(H)=\sigma_{\mathsf{ess}}
(H_+)\subset\sigma(H_+)\subset[0,\infty)$ and $\sigma_{\mathsf{ess}}(H_-)=\sigma_{\mathsf{ess}}(H_0)=\sigma(H_0)=
[0,\infty)$.

For any potential function $W$ verifying the same conditions as $V$ above, we denote by $N(W)$ the number of
strictly negative eigenvalues (counted with their multiplicity) of the operator $H_0\dotplus W$. The following
result reduces our study to the case $V_+=0$.

\begin{lemma}
 The following inequality is true:
$$
N(V)\ \leq\ N(-V_-).
$$
In particular we have that $N(V)=\infty$ implies that $N(-V_-)=\infty$.
\end{lemma}

\begin{proof}
 We apply the {\sl Min-Max} principle (see Theorem XIII.2 in \cite{RS}) noticing that
 $\mathcal{D}(\mathfrak{h_-})=\mathcal{D}(\mathfrak{h_0})\supset\mathcal{D}(\mathfrak{h})$ and
 $\mathfrak{h}_-\leq\mathfrak{h}$ and we deduce that the operator $H_-$ has at least $N(V)$ strictly negative
 eigenvalues.
\end{proof}

\noindent Thus we shall suppose from now on that $V_+=0$.

\subsection{Reduction to smooth, compactly supported potentials}

In this subsection we shall prove that we can suppose $V_-\in C_0^\infty(\Rd)$. This will be done by approximation,
using a result of the type of Theorem 4.1 from \cite{S3}.

\begin{lemma}\label{pr}
Let $V$ and $V_n$ ($n\ge 1$) functions as in proposition \ref{Ham-rel-pert}. In addition, $V_+=V_{n,+}=0$ for all $n\ge 1$ and
$\lim_{n\rightarrow\infty}V_{n,-}=V_-$ in $L^1_{\rm{loc}}(\Rd)$ and $V_{n,-}$ are uniformly $\mathfrak h_0$-bounded with relative
bound $<1$. We set $H_n:=H_A\dotplus V_n$. Then $H_n\rightarrow H$ when $n\rightarrow\infty$ in strong resolvent sense.
\end{lemma}

\begin{proof}
We denote by $\mathfrak h_n$ the quadratic form associated to $H_n$, i.e.
$\mathfrak h_n=\mathfrak h_A-\mathfrak q_{n,-}$, where $\mathfrak q_{n,-}$ is associated to
$V_{n,-}$ by (\ref{qformV}). We have $D(h_n)=D(h_A)\subset D(q_{n,-})$, and according to Proposition \ref{Ham-rel-pert} there exist
$\alpha\in (0,1)$ and $\beta>0$ such that

\begin{equation}\label{forme}
\mathfrak q_{n,-}(v)\le \alpha \mathfrak h_A(v)+\beta\parallel v\parallel,
\ \ \forall v\in D(\mathfrak h_A),\ \forall n\ge 1.
\end{equation}

It follows that $\mathfrak h_n$ are uniformly lower bounded and the norms defined on $D(\mathfrak h_A)$ by
$\mathfrak h_A$ and $\mathfrak h_n$ are
equivalent, uniformly with respect to $n\ge 1$. Moreover, $C_0^\infty(\Rd)$ is a core for $H_A$, thus for $\mathfrak h_A$,
$\mathfrak h$ and $\mathfrak h_n$ also.

Let $f\in L^2(\Rd)$ and $u_n:=(H_n+i)^{-1}f\in D(H_n)\subset D(\mathfrak h_A)$, $n\ge 1$. We have clearly

\begin{equation}\label{afara}
\parallel u_n\parallel\le\parallel f\parallel,\ \ |\mathfrak h_n(u_n)|=|(H_nu_n,u_n)|\le\parallel f\parallel,
\ \ \forall n\ge 1.
\end{equation}

From (\ref{forme}), the subsequent comments and (\ref{afara}) it follows that the sequence $(u_n)_{n\ge 1}$
is bounded in $D(\mathfrak h_A)$, while the sequence $\l(V_{n,-}^{1/2}u_n\r)_{n\ge 1}$ is bounded in $L^2(\Rd)$.
Let $u\in L^2(\Rd)$ be a limit point of the sequence $(u_n)_{n\ge 1}$ with respect to the weak topology on $L^2(\Rd)$. 
By restricting maybe to a subsequence,
we may assume that there exist $\psi,\eta\in L^2(\Rd)$ such that $H_A^{1/2}u_n\underset{n\rightarrow\infty}{\rightarrow}\psi$ and
$V_{n,-}^{1/2}u_n\underset{n\rightarrow\infty}{\rightarrow}\eta$ in
the weak topology of $L^2(\Rd)$.
For
$g\in D\l(H_A^{1/2}\r)$ we have
$$
\l(H_A^{1/2}g,u\r)=\lim_{n\rightarrow\infty}\l(H_A^{1/2}g,u_n\r)=\lim_{n\rightarrow\infty}\l(g,H_A^{1/2}u_n\r)=
(g,\psi),
$$
thus $u\in D(H_A^{1/2})$ and $H_A^{1/2}u=\psi$. Then $u\in D(\mathfrak q_-)$ and for any $g\in C^\infty_0(\Rd)$
$$
(\eta,g)=\lim_{n\rightarrow \infty}\l(V_{n,-}^{1/2}u_n,g\r)=\lim_{n\rightarrow \infty}\l(u_n,V_{n,-}^{1/2}g\r)=
\l(u,V_{-}^{1/2}g\r)=\l(V_{-}^{1/2}u,g\r),
$$
implying $V_-^{1/2}u=\eta$.

It follows that for every $g\in C^\infty_0(\Rd)$ we have
$$
(g,f)=(g,(H_n+i)u_n)=\mathfrak h_n(g,u_n)-i(g,u_n)=
$$
$$
=\l(H_A^{1/2}g,H_A^{1/2}u_n\r)-\l(V_{n,-}^{1/2}g,V^{1/2}_{n,-}u_n\r)
-i(g,u_n)\rightarrow \mathfrak h(g,u)-i(g,u).
$$
Consequently, $u\in D(H)$ and $(H+i)u=f$. Thus the sequence $(u_n)_{n\ge 1}$ has the single limit point
$u=(H+i)^{-1}f$ for the weak topology of $L^2(\Rd)$. It follows that $(H_n\pm i)^{-1}f\rightarrow
(H\pm i)^{-1}f$ weakly in $L^2(\Rd)$ for $n\rightarrow\infty$.

By the resolvent identity we get
$$
\parallel(H_n+i)^{-1}f\parallel^2=\frac{i}{2}\l((f,(H_n-i)^{-1}f)-(f,(H_n+i)^{-1}f)\r)\rightarrow
\parallel(H+i)^{-1}f\parallel^2,
$$
therefore $(H_n+i)^{-1}f\rightarrow (H+i)^{-1}f$ in $L^2(\Rd)$.
\end{proof}

A direct consequence of Lemma \ref{pr} and Theorem VIII.20 from \cite{RS} is

\begin{corollary}
Under the hypothesis of Lemma \ref{pr}, for any function $f$ bounded and continuous on $\mathbb R$ and
any $u\in L^2(\Rd)$, we have $f(H_n)u\rightarrow f(H)u$.
\end{corollary}

Approximating $V_-$ is done by the standard procedures: cutoffs and regularization. The first of the lemmas
below is obvious.

\begin{lemma}\label{cut-off-a}
 Let $V_-\in L^1_{\mathsf{loc}}(\Rd)$ with $V_-\geq0$ and assume that its associated sesquilinear form is
 $\mathfrak{h}_0$-bounded with relative bound strictly less then 1. Let $\theta\in C_0^\infty([0,\infty))$
 satisfy the following: $0\leq\theta\leq1$, $\theta$ is a decreasing function, $\theta(t)=1$ for
 $t\in[0,1]$ and $\theta(t)=0$ for $t\in[2,\infty)$.

 If we denote by $\theta^{n}(x):=\theta(|x|/n)$ and
 $V_-^{n}=\theta^{n}V_-$, then $V_-^{n}\rightarrow V_-$ in $L^1_{\mathsf{loc}}(\Rd)$,
 $0\leq V_-^{n}\leq V_-^{n+1}$ and the sesquilinear forms associated to
 $V_-^{n}$ are $\mathfrak{h}_0$-bounded with relative bound strictly less then 1, uniformly in $n\in\mathbb{N}^*$,
 .

 Moreover, if we denote by $\mathfrak{h}^{n}$ the sesquilinear
 form associated to the operator $H_A\dotplus (-V_-^{n})$, we have
 $\mathfrak{h}^{(n)}\geq\mathfrak{h}^{(n+1)}\geq\mathfrak{h}$ and $\mathfrak{h}^{(n)}(u)
 \underset{n\rightarrow\infty}{\rightarrow}\mathfrak{h}(u)$ for any $u\in\mathcal{D}(\mathfrak{h}_A)$.

 If, in addition, $V_-\in L^p(\Rd)$, $p\ge 1$, then $V^n_-\in L^p_{\rm{comp}}(\Rd)$,
 $\| V^n_-\|_{L^p}\le\| V_-\|_{L^p}$ for any $n\ge 1$, and $V^n_-\rightarrow V_-$
 in $L^p(\Rd)$.
\end{lemma}

\begin{lemma}\label{ultima}
(a) Let $V_-\in L^1_{\rm{loc}}(\Rd)$, $V_-\ge 0$ and $\mathfrak h_0$-bounded with relative bound $<1$.
Let $\theta\in C^\infty_0(\Rd)$, $\theta\ge 0$ and $\int_{\Rd}\theta=1$. We set $\theta_n(x):=n^{d}\theta(nx)$,
$x\in\Rd$, $n\in \mathbb N^*$ and $V_{n,-}:=V_-\ast \theta_n\in C^\infty_0$. In particular, $V_{n,-}\in C^\infty_0(\Rd)$
if $V_-\in L^1_{\rm{comp}}(\Rd)$.

Then $V_{n,-}\rightarrow V_-$ in $L^1_{\rm{loc}}(\Rd)$ for $n\rightarrow\infty$ and the functions $V_{n,-}$ are non-negative and uniformly
$h_0$-bounded, with relative bound $<1$. Moreover, $\mathfrak h_n(u)\rightarrow \mathfrak h(u)$
for any $u\in D(\mathfrak h_A)$, where $\mathfrak h_n$
is the quadratic form associated to $H_n:=H_A\overset{\cdot}{+}(-V_n)$.

(b) If, in addition, $V_-\in L^p(\Rd)$ with $p\ge 1$, then $V_{n,-}\in L^p(\Rd)\cap C^\infty(\Rd)$,
$\parallel V_{n,-}\parallel_{L^p}\le\parallel V_-\parallel_{L^p}$, $\forall n\ge 1$ and $V_{n,-}\rightarrow V_-$
in $L^p(\Rd)$.
\end{lemma}

\begin{proof}
(a) We have for any $x\in\Rd$
\begin{equation}\label{needed}
V_{n,-}(x)=\int_{\Rd}dy\,\theta_n(y)V_-(x-y)=\int_{\Rd}dy\,\theta(y)V_-(x-n^{-1}y).
\end{equation}

By the Dominated Convergence Theorem, for any compact $K\subset \Rd$
$$
\int_{K}dx\,\vert V_{n,-}(x)-V_-(x)\vert\le\int_{\Rd}dy\,\theta(y)\int_K dx\,\vert V_-(x-n^{-1}y)-V_-(x)\vert
\rightarrow 0,
$$
hence $V_{n,-}$ converges to $V_-$ in $L^1_{\rm{loc}}(\Rd)$ when $n\rightarrow\infty$.

If $V_-$ is relatively small with respect to $\mathfrak{h}_0$, we use the fact that $H_0^{1/2}$ is a convolution operator (hence it commutes with
translations) and using the comments after inequality (\ref{forme}), we deduce that for any $u\in C_0^\infty(\Rd)$ there exists $\alpha\in(0,1)$ and $\beta\geq0$ such that
$$
\int_{\Rd}dx\,V_{n,-}|u|^2=\int_{\Rd}dy\,\theta_n(y)\int_{\Rd}dz\, V_-(z)|u(z+y)|^2\le
$$
$$
\le\int_{\Rd}dy\,\theta_n(y)\l[\alpha\parallel H_0^{1/2}u(\cdot+y)\parallel^2+\beta
\parallel u(\cdot+y)\parallel^2\r]=
$$
$$
=\alpha\parallel H_0^{1/2}u\parallel^2
+\beta\parallel u\parallel^2.
$$

(b) From (\ref{needed}) it follows that
$$
\parallel V_{n,-}\parallel_{L^p}\le\int_{\Rd}dy\,\theta_n(y)\parallel V_-(\cdot-y)
\parallel_{L^p}\leq\parallel V_-\parallel_{L^p}.
$$
Also, using the Dominated Convergence Theorem, we infer that
$$
\parallel V_{n,-}-V_-\parallel_{L^p}\le\int_{\Rd}dy\,\theta(y)\parallel V_-(\cdot)-V_-(\cdot-n^{-1}y)
\parallel_{L^p}\rightarrow 0.
$$
\end{proof}

Thus  Lemmas \ref{cut-off-a} and \ref{ultima} imply, for a potential function $V_-$ satisfying the hypothesis of the Lemma, the existence of a sequence $(V_{n,-})_{n\ge 1}
\subset C^\infty_0(\Rd)$ such that $V_{n,-}\ge 0$, $\parallel V_{n,-}\parallel_{L^p}\le\parallel V_{-}
\parallel_{L^p}$, $\forall n\ge 1$, $V_{n,-}\rightarrow V_-$ in $L^p(\Rd)$ (for $p=d$ and $p=d/2$) when $n\rightarrow\infty$ and
the functions $V_{n,-}$ are uniformly $\mathfrak h_0$-bounded with relative bound $<1$.

\begin{lemma}\label{saturat}
Assume that there exists a constant $C>0$, such that the inequality
\begin{equation}
N(-V_{n,-})\le C\left(\int_{\Rd}dx\,|V_{n,-}(x)|^d\,+\,\int_{\Rd}dx\,|V_{n,-}(x)|^{d/2}\right)
\end{equation}
holds for any $n\ge 1$. Then one also has
\begin{equation}
N(-V_{-})\le C\left(\int_{\Rd}dx\,|V_{-}(x)|^d\,+\,\int_{\Rd}dx\,|V_{-}(x)|^{d/2}\right).
\end{equation}
\end{lemma}
\begin{proof}
We set $H_{n,-}:=H_0\dotplus (-V_{n,-})$; $(E_{n,-}(\lambda))_{\lambda\in\mathbb R}$ will be the spectral
family of $H_{n,-}$ and $(E_{-}(\lambda))_{\lambda\in\mathbb R}$ the spectral
family of $H_{-}$. For $\lambda<0$, we denote by $N_\lambda(W)$ the number of eigenvalues of $H_0\dotplus W$
which are strictly smaller than $\lambda$ (for any potential function $W$ satisfying the hypothesis at the begining of this section). It suffices to show that for any $\lambda<0$ not belonging to
the spectrum of $H_-$, one has the inequality
\begin{equation}\label{gar}
N_\lambda(-V_{-})\le C\left(\int_{\Rd}dx\,|V_{-}(x)|^d\,+\,\int_{\Rd}dx\,|V_{-}(x)|^{d/2}\right).
\end{equation}
Since $V_{n,-}$ converges to $V_-$ in $L^1_{\rm{loc}}(\Rd)$, cf. Lemma \ref{pr}, $H_{n,-}$ will converge to $H_-$
in strong resolvent sense. By \cite{K}, Ch.VIII, Th.1.15, this implies the strong convergence of $E_{n,-}(\lambda)$
to $E_-(\lambda)$ for any $\lambda\notin\sigma(H_-)$. By Lemmas 1.23 and 1.24 from \cite{K}, Ch.VII, for $\lambda<0$,
$\lambda\notin \sigma(H_-)$, one also has $\parallel E_{n,-}(\lambda)-E_-(\lambda)\parallel\rightarrow 0$. 
Let us suppose that there exists some $\lambda<0$ not belonging to $\sigma(H_-)$ and such that for it the inequality (\ref{gar}) is not verified. Thus for the given $\lambda<0$ we have $\forall n\geq 1$:
$$
N(-V_{n,-})\le C\left(\int_{\Rd}dx\,|V_{-}(x)|^d\,+\,\int_{\Rd}dx\,|V_{-}(x)|^{d/2}\right)< N_\lambda(-V_{-}).
$$
But for $n$ large enough, one has
$
N_\lambda(-V_-)=N_{\lambda}(-V_{n,-})
$
and thus
$$
N_\lambda(-V_-)=N_{\lambda}(-V_{n,-})\,\le\, N(-V_{n,-})\le
$$
$$
\leq C\left(\int_{\Rd}dx\,|V_{n,-}(x)|^d\,+\,\int_{\Rd}dx\,|V_{n,-}(x)|^{d/2}\right)\le
$$
$$
\le C\left(\int_{\Rd}dx\,|V_{-}(x)|^d\,+\,\int_{\Rd}dx\,|V_{-}(x)|^{d/2}\right)
$$
that is a contradiction with our initial hypothesis.
\end{proof}

\subsection{Proof of the Theorem \ref{Main} for $\boldsymbol{B=0}$}

We shall assume from now on that $V_+=0$ and $0\le V_-\in C^\infty_0(\Rd)$. We check a Birman-Schwinger principle.
For $\alpha>0$ we set $K_\alpha:=V_-^{1/2}(H_0+\alpha)^{-1}V_-^{1/2}$; it is a positive compact operator
on $L^2(\Rd)$.

\begin{lemma}\label{olema}
\begin{equation}\label{olem}
N_{-\alpha}(-V_-)\le\;\card\{\mu>1\mid \mu \ {\rm eigenvalue\  of}\  K_\alpha\}.
\end{equation}
\end{lemma}

\begin{proof}
We introduce the sequence of functions $\mu_n:[0,\infty)\rightarrow (-\infty,0]$, $n\ge 1$, where $\mu_n(\lambda)$ is
the n'th eigenvalue of $H_0-\lambda V_-$ if this operator has at least $n$ strictly negative eigenvalues and
$\mu_n(\lambda)=0$ if not. Cf. \cite{RS} \S XIII.3, $\mu_n$ is continuous and decreasing (even strictly
decreasing on intervals on which it is strictly negative). Obviously, we have
$N_{-\alpha}(-V_-)\le\,\card\{n\ge 1\mid \mu_n(1)<-\alpha\}$. Now fix some $n$ such that $\mu_n(1)<-\alpha$ and recall
that $\mu_n(0)=0$. The function $\mu_n$ is continuous and injective on the interval $[\epsilon_n,1]$,
where $\epsilon_n:=\sup\{\lambda\ge 0\mid \mu_n(\lambda)=0\}$, therefore it exists a unique $\lambda\in(0,1)$
such that $\mu_n(\lambda)=-\alpha$. Thus
$$
N_{-\alpha}(-V_-)=\;\card\{\lambda\in(0,1)\mid \exists n\ge 1\ s.t.\ \mu_n(\lambda)=-\alpha\}=
$$
$$
=\;\card\{\lambda\in(0,1)\mid \exists \varphi\in D(H_0)\setminus\{0\}\ s.t.\ (H_0-\lambda V_-)\varphi=
-\alpha\varphi\}\le
$$
$$
\le\;\card\{\lambda\in(0,1)\mid \exists \psi\in L^2(\Rd)\setminus\{0\}\ s.t.\ K_\alpha\psi=\lambda^{-1}\psi\},
$$
where for the last inequality we set $\psi:=V_-^{1/2}\varphi$, noticing that the equality
$(H_0+\alpha)\varphi=\lambda V_-\varphi$ implies $\psi\ne 0$.
\end{proof}

\begin{lemma}\label{dolema}
Let $F:[0,\infty)\rightarrow[0,\infty)$ be a strictly increasing continuous function with $F(0)=0$.
Then $F(K_\alpha)$ is a positive compact operator and the next inequality holds:
$$
N_{-\alpha}(-V_-)\le F(1)^{-1}\sum_{F(\mu)\in\sigma[F(K_\alpha)], F(\mu)> F(1)}F(\mu).
$$
\end{lemma}
\begin{proof}
The first part is obvious. Using (\ref{olem}) and $F$'s monotony, we get
$$
N_{-\alpha}(-V_-)\le\sharp\{\mu>1\mid\mu\in\sigma(K_\alpha)\}=
\;\card\{F(\mu)\mid\mu>1, F(\mu)\in\sigma[F(K_\alpha)]\}=
$$
$$
\sum_{\mu>1, F(\mu)\in\sigma[F(K_\alpha)]}\frac{F(\mu)}{F(\mu)}\le
F(1)^{-1}\sum_{\mu>1, F(\mu)\in\sigma[F(K_\alpha)]}F(\mu).
$$
\end{proof}
So, we shall be interested in finding functions $F$ having the properties in the statement above, such that
$F(K_\alpha)\in B_1$ (the ideal of trace-class operators in $L^2(\Rd)$)
and such that $\rm{Tr}\l[F(K_\alpha)\r]$ is conveniently estimated.

Using an idea from \cite{S1}, we are going to consider functions of the form
$$
F(t):=t\int_0^{\infty}ds\,e^{-s}g(ts),\ \ t\ge 0,
$$
where $g:[0,\infty)\rightarrow[0,\infty)$ is continuous, bounded and $g\equiv\!\!\!\!\!\!\!\!\diagup\, 0$. Plainly,
$F:[0,\infty)\rightarrow[0,\infty)$ is continuous, $F(0)=0$, satisfies $F(t)\le Ct$ for some $C>0$ and the identity
$$
F(t)=\int_0^{\infty}dr\,e^{-rt^{-1}}g(r)
$$
implies that $F$ is strictly increasing. We shall use the notations $F=\Phi(g)$, $\tilde{g}(t):=tg(t)$.

In particular, $g_\lambda(t)=e^{-\lambda t}$,
$\lambda>0$ leads to $F_\lambda(t)=t(1+\lambda t)^{-1}$. In the sequel,
relations valid for this particular case will be extended to the following case, that we shall be interested in:
\begin{equation}\label{util}
g_\infty:[0,\infty)\rightarrow[0,\infty),\ \  g_\infty(t)=0\ {\rm if} \ 0\le t\le 1, \ \ g_\infty(t)=1-1/t\ {\rm
if}\ t>1,
\end{equation}
by using an approximation that we now introduce. The first lemma is obvious.
\begin{lemma}\label{la}
Let $g_\infty$ be given by (\ref{util}). For $n\ge 1$ we define $g_n:[0,\infty)\rightarrow[0,1]$,
$g_n(t)=g(t)$ for $0\le t\le n$, $g_n(t)=\frac{2n-1}{t}-1$ for $n\le t\le 2n-1$, $g_n(t)=0$ for $t\ge 2n-1$.
Then $g_n\in C_0((0,\infty))$, $0\le g_n\le g_{n+1}\le g_\infty$, $\forall n$ and $g_n\rightarrow g_\infty$ when
$n\rightarrow\infty$ uniformly on any compact subset of $[0,\infty)$.
\end{lemma}

\begin{lemma}\label{le}
Let $f$ be a nonnegative continuous function on $[0,\infty)$, $\lim_{t\rightarrow\infty}f(t)=0$.
There exists a sequence $(f^k)_{k\ge 1}$ of real functions on $[0,\infty)$ with the properties

(a) Every $f^k$ is a finite linear combination of functions of the form $g_\lambda$, $\lambda>0$.

(b) $f^k\ge f^{k+1}\ge f\geq0$ on $[0,\infty)$, $\forall k\ge 1$,

(c) $f^k\rightarrow f$ uniformly on $[0,\infty)$ when $k\rightarrow\infty$.
\end{lemma}
\begin{proof}
We define the function $h:[0,1]\rightarrow[0,\infty)$, $h(s):=f(-{\rm ln}s)$ for $s\in(0,1]$, $h(0):=0$.
It follows that $h\in C([0,1])$. We can chose now two sequences of positive numbers $\{\epsilon_k\}_{k\geq 1}$ and $\{\delta_k\}_{k\geq 1}$ verifying the properties: $\underset{k\rightarrow\infty}{\lim}(\epsilon_k+\delta_k)=0$ and $\delta_k-\epsilon_k\geq\epsilon_{k+1}+\delta_{k+1}>0,\forall k\geq 1$ (for example we may take $\delta_k=(k+2)^{-1}$ and $\epsilon_k=(k+2)^{-3}$). Using the Weierstrass Theorem we may find for any $k\geq 1$ a real polynomial $P^\prime_k$ such that $\underset{s\in[0,1]}{\sup}|h(s)-P^\prime_k(s)|\leq\epsilon_k$ and let us denote by $P_k:=P^\prime_k+\delta_k$. We get:
$$
\underset{s\in[0,1]}{\sup}|h(s)-P_k(s)|\leq\epsilon_k+\delta_k\underset{k\rightarrow\infty}{\rightarrow}0,
$$
$$
h\leq h+\delta_{k+1}-\epsilon_{k+1}\leq P^\prime_{k+1}+\delta_{k+1}=P_{k+1}\leq h+\delta_{k+1}+\epsilon_{k+1}\leq
$$
$$
\leq h+\delta_{k}-\epsilon_{k}\leq P^\prime_k+\delta_k=P_k
$$
on $[0,1]$. Thus $f^k(t):=P_k(e^{-t})$ defined on $[0,\infty)$ for $k\ge 1$
have the required properties.
\end{proof}

\begin{proposition}\label{l}
Let $F_\infty:=\Phi(g_\infty)$. The operator $F_\infty(K_\alpha)$ is self-adjoint, positive and compact on $L^2(\Rd)$.
It admits an integral kernel of the form
\begin{equation}\label{f1}
\l[F_\infty(K_\alpha)\r](x,y)=
\end{equation}
\begin{equation*}
=V_-^{1/2}(x)V_-^{1/2}(y)\int_0^\infty dt\,e^{-\alpha t}\int_\Omega \mu_{0,x}^{t,y}(d\omega)
g_\infty\l(\int_0^t ds\,V_-(X_s)\r),
\end{equation*}
which is continuous, symmetric, with $\l[F_\infty(K_\alpha)\r](x,x)\ge 0$.
\end{proposition}

\begin{proof}
The first part is clear. To establish (\ref{f1}), we treat first the operator $B_\lambda:=F_\lambda(K_\alpha)$,
$\lambda>0$. We have
\begin{equation}\label{be}
B_\lambda=K_\alpha(1+\lambda K_\alpha)^{-1}\implies B_\lambda=K_\alpha-\lambda B_\lambda K_\alpha.
\end{equation}
The second resolvent identity gives
$$
(H_0+\alpha)^{-1}-(H_0+\lambda V_-+\alpha)^{-1}=\lambda(H_0+\lambda V_-+\alpha)^{-1}V_-(H_0+\alpha)^{-1}.
$$
Multiplying by $V_-^{1/2}$ to the left and to the right and taking into account (\ref{be}) and the definition
of $K_\alpha$, one gets
$$
B_\lambda=V_-^{1/2}(H_0+\lambda V_-+\alpha)^{-1}V_-^{1/2}=V_-^{1/2}\l[\int_0^\infty dt\,
e^{-\alpha t}e^{-t(H_0+\lambda V_-)}\r]V_-^{1/2}.
$$
By Proposition \ref{FK} and its consequences, for any $u\in C_0(\Rd)$, $u\ge 0$, we have
\begin{equation}\label{lambda}
\l[F_\lambda(K_\alpha)u\r](x)=
\end{equation}
\begin{equation*}
=V_-^{1/2}(x)\int_0^\infty dt\,
e^{-\alpha t}\int_{\Rd} dy\,\l[\int_\Omega\,\mu_{0,x}^{t,y}(d\omega)\,g_\lambda\l(\int_0^t ds\,V_-(X_s)\r)\r]
V_-^{1/2}(y)u(y).
\end{equation*}
Since $\Phi$ maps monotonous convergent sequences into monotonous convergent sequences, by applying Lemmas \ref{la}
and \ref{le} and the Monotonous Convergence Theorem (B. Levi), we get (\ref{lambda}) for $\lambda=\infty$, for the
couple $(g_\infty,F_\infty)$.

We introduce the notation
\begin{equation}
G_\lambda(t;x,y):=\int_\Omega\,\mu_{0,x}^{t,y}(d\omega)\,g_\lambda\l(\int_0^t ds\,V_-(X_s)\r),
\ \ t>0,\;x,y\in\Rd, \;0<\lambda\le\infty.
\end{equation}
By the consequences of Proposition \ref{FK}, for any $0<\lambda<\infty$ the function $G_\lambda$ is
continuous on $(0,\infty)\times \Rd
\times\Rd$ and symmetric in $x,y$. To obtain the same properties for $\lambda=\infty$, we approximate $g_\infty$
by using once again Lemmas \ref{la} and \ref{le}. So it exists a sequence $(f_n)_{n\ge 1}$ of real continuous
functions on $[0,\infty)$, each one being a finite linear combination of functions of the form
$g_\lambda$, such that $f_n$ converges to $g_\infty$ uniformly on any compact subset of $[0,\infty)$.
On the other hand, if $M>0$ is an upper bound for $V_-$, we have
$$
0\le\int_0^t ds\,V_-(X_s)\le Mt,
$$
and $\mu_{0,x}^{t,y}(\Omega)=\overset{\circ}{\wp}_t(x-y)$. It follows that $G_\infty$ is, uniformly on compact subsets of
$[0,\infty)\times\Rd\times\Rd$, the limit of a sequence of continuous functions, which are symmetric in $x,y$.
Thus $G_\infty$ has the same properties. Moreover, since $0\le g_\infty\le 1$ and $g_\infty(t)=0$ for $0\le t\le 1$,
we have $G_\infty(t;x,y)=0$ for $t\le 1/M$. Using (\ref{iato}) and (\ref{iata}), there is a constant $C>0$
such that
\begin{equation}\label{ignore}
0\le G_\infty(t;x,y)\le C,\ \ \forall t>0,\ \forall x,y\in\Rd.
\end{equation}
From (\ref{lambda}) for $\lambda=\infty$, we infer that $F_\infty(K_\alpha)$ has an integral kernel of the form
\begin{equation}\label{ik}
\l[F_\infty(K_\alpha)\r](x,y)=V_-^{1/2}(x)V_-^{1/2}(y)\int_0^\infty dt\,e^{-\alpha t}G_\infty(t;x,y),
\end{equation}
so (\ref{f1}) is verified. The continuity of $F_\infty(K_\alpha)$ follows from the Dominated Convergence Theorem
and from (\ref{ignore}). The symmetry is obvious, and the last property of the statement follows from
$F_\infty(K_\alpha)\ge 0$.
\end{proof}

\begin{remark}
By a lemma from \cite{RS}, \S XI.4, $F_\infty(K_\alpha)\in B_1$ if the function
$\Rd\ni x\mapsto \l[F_\infty(K_\alpha)\r](x,x)$
is integrable and one has
\begin{equation}\label{once}
\rm{Tr}\l[F_\infty(K_\alpha)\r]=\int_{\Rd} dx\,\l[F_\infty(K_\alpha)\r](x,x).
\end{equation}
Setting $D_\infty(t;x):=V_-(x)G_\infty(t;x,x)$, $t>0,x\in\Rd$, we have
\begin{equation}\label{doce}
\l[F_\infty(K_\alpha)\r](x,x)=\int_0^\infty dt\,e^{-\alpha t}D_\infty(t;x).
\end{equation}
\end{remark}
To check the integrability of this function, one introduces 
$$
\Psi_\infty:(0,\infty)\times \Rd\rightarrow \mathbb R_+,
$$
$$
\Psi_\infty(t;x):= t^{-1}\int_\Omega\mu^{t,x}_{0,x}(d\omega)\, \tilde{g}_\infty\l(\int_0^t ds\,V_-(X_s)\r),
$$
where $\tilde{g}_\infty(t):=tg_\infty(t)$. The role of this function is stressed by

\begin{lemma}\label{if}
For $d\geq3$ consider the following constant depending only on $d$:
$$
\overline{C}_d:=C\left(\int_1^\infty ds\,s^{-d}\,g_\infty(s)\,\vee\int_1^\infty ds\,s^{-d/2}\,g_\infty(s)\right)=C\int_1^\infty ds\,s^{-d/2}\,g_\infty(s)
$$
where $C$ is the constant verifying (\ref{vine}).
One has
\begin{equation}
\int_0^\infty dt\,e^{-\alpha t}\int_{\Rd} dx\,\Psi_\infty(t;x)\le \overline{C}_d\left(\int_{\Rd} dx\,V_-^d(x)+\int_{\Rd} dx\,V_-^{d/2}(x)\right).
\end{equation}
\end{lemma}
\begin{proof}
The function $\tilde{g}_\infty$ is convex and $\frac{ds}{t}$ is a probability on $(0,t)$; thus by the Jensen inequality we obtain
$$
\tilde{g}_\infty\left(\int_0^t ds\,V_-(X_s)\right)\,\leq\,\int_0^t \frac{ds}{t}\,\tilde{g}_\infty\left(t\,V_-(X_s)\right).
$$
Let us also remark that for the constant $\overline{C}_d$ to be finite we have to ask that $d\geq3$ for the factor $s^{-d/2}$ to be integrable at infinity, because the convexity condition on $\tilde{g}_\infty$ rather implies that $g_\infty$ cannot vanish at infinity.

Then
$$
\int_0^\infty dt\,e^{-\alpha t}\int_{\Rd} dx\,\Psi_\infty(t;x)\le 
$$
$$
\leq\int_0^\infty dt\,t^{-2}\,e^{-\alpha t}\int_{\mathbb{R}^d}dx\,\left[\int_{\Omega}\mu_{0,x}^{t,x}(d\omega)\int_0^t ds\,\tilde{g}_\infty\left(tV_-(X_s)\right)\right].
$$
Using now Proposition \ref{4.3}, the last expression is equal to:
$$
\int_0^\infty dt\,t^{-2}\,e^{-\alpha t}\int_{\mathbb{R}^d}dx\,\left[\int_{\Omega}\mu_{0,0}^{t,0}(d\omega)\int_0^t ds\,\tilde{g}_\infty\left(tV_-(x+\omega(s))\right)\right]=
$$
$$
=\int_0^\infty dt\,t^{-2}\,e^{-\alpha t}\left[\int_{\Omega}\mu_{0,0}^{t,0}(d\omega)\int_0^t ds\int_{\mathbb{R}^d}dx\,\tilde{g}_\infty\left(tV_-(x)\right)\right]=
$$
$$
=\int_0^\infty dt\,t^{-1}\,e^{-\alpha t}\left[\int_{\Omega}\mu_{0,0}^{t,0}(d\omega)\right]\int_{\mathbb{R}^d}dx\,\tilde{g}_\infty\left(tV_-(x)\right)=
$$
$$
=\int_0^\infty dt\,t^{-1}\,e^{-\alpha t}\overset{\circ}{\wp}_t(0)\int_{\mathbb{R}^d}dx\,\tilde{g}_\infty\left(tV_-(x)\right)\leq
$$
$$
\leq C\int_{\mathbb{R}^d}dx\,\left[\int_0^\infty dt\,t^{-d-1}(1+t^{d/2})\tilde{g}_\infty\left(tV_-(x)\right)\right]\leq
$$
$$
\leq \overline{C}_d\left(\int_{\Rd} dx\,V_-^d(x)+\int_{\Rd} dx\,V_-^{d/2}(x)\right),
$$
where we have used the fact that $s<1$ implies $g_\infty(s)=0$.
\end{proof}

The next result gives the connection between $D_\infty$ and $\Psi_\infty$:
\begin{proposition}
$$
 \int_{\mathbb{R}^d}dx\,D_\infty(t,x)\,=\,\int_{\mathbb{R}^d}dx\,\Psi_\infty(t,x).
$$
\end{proposition}
\begin{proof}
 First let us verify the following identity for any $t>0$:
\begin{equation}\label{5.21}
 \int_{\mathbb{R}^d}dx\,D_\lambda(t,x)\,=\,\int_{\mathbb{R}^d}dx\,\Psi_\lambda(t,x), \quad\text{for }\lambda\in(0,\infty)
\end{equation}
where $D_\lambda$ and $\Psi_\lambda$ are defined in terms of $g_\lambda$ in the same way that $D_\infty$ and $\Psi_\infty$ are defined in terms of $g_\infty$. Let us point out that both $D_\lambda$ and $\Psi_\lambda$ are positive measurable functions on $(0,\infty)\times\mathbb{R}^d$ but only the integral on the left hand side of (\ref{5.21}) is evidently finite by what we have proven so far. For simplifying the writing we shall take $\lambda=1$. For any $r\in[0,t]$ we denote by
$$
S_r:=e^{-r(H_0+V_-)}V_-e^{-(t-r)(H_0+V_-)}.
$$
Following the remarks after Proposition \ref{FK} above, for $r\in(0,t)$, both exponentials appearing in the above right hand side are integral operators with non-negative continuous integral kernels; thus $S_r$ will also be an integral operator with non-negative continuous kernel that we shall denote by $K_r$, and we can compute it explicitely as follows. For a non-negative $u\in C_0(\mathbb{R}^d)$, using Proposition \ref{FKI} with $A=0$ gives
$$
(S_ru)(x)=\mathsf{E}_x\left\{e^{-\int_0^r V_-(X_\rho)d\rho}V_-(X_r)\mathsf{E}_{X_r}\left[e^{-\int_0^{t-r} V_-(X_\sigma)d\sigma}u(X_{t-r}) \right] \right\}
$$
and using the Markov property (\ref{Markov}) we obtain
$$
\mathsf{E}_{X_r}\left[e^{-\int_0^{t-r} V_-(X_\sigma)d\sigma}u(X_{t-r}) \right]=\mathsf{E}_{x}\left[e^{-\int_0^{t-r} V_-(X_\sigma\circ\theta_r)d\sigma}u(X_{t})\mid\mathfrak{F}_r \right]=
$$
$$
=\mathsf{E}_{x}\left[e^{-\int_r^{t} V_-(X_\sigma)d\sigma}u(X_{t})\mid\mathfrak{F}_r \right].
$$
As the function $e^{-\int_0^r V_-(X_\rho)d\rho}V_-(X_r):\Omega\rightarrow\mathbb{R}$ is evidently $\mathfrak{F}_r$-measurable, we get (using the property (\ref{cond-expect-prod}) of conditional expectations)
$$
(S_ru)(x)=\mathsf{E}_{x}\left\{\mathsf{E}_{x}\left(V_-(X_r)e^{-\int_0^{t} V_-(X_\sigma)d\sigma}u(X_{t})\mid\mathfrak{F}_r\right) \right\}.
$$
We use now the property (\ref{expect-cond-expect}) and Proposition \ref{FK} taking $F:=V_-(X_r)$ in order to get
$$
(S_ru)(x)=\mathsf{E}_{x}\left\{V_-(X_r)e^{-\int_0^{t} V_-(X_\sigma)d\sigma}u(X_{t}) \right\}=
$$
$$
=\int_{\mathbb{R}^d}dy\,\left\{\int_\Omega \mu^{t,y}_{0,x}(d\omega)V_-(X_r)e^{-\int_0^{t} V_-(X_\sigma)d\sigma} \right\}u(y).
$$
In conclusion for any $(x,y)\in\mathbb{R}^d\times\mathbb{R}^d$ we have
\begin{equation}\label{5.24}
 K_r(x,y)=\int_\Omega \mu^{t,y}_{0,x}(d\omega)V_-(X_r)e^{-\int_0^{t} V_-(X_\sigma)d\sigma} .
\end{equation}
Using Proposition \ref{4.3} we obtain
$$
\int_{\mathbb{R}^d}dx\,K_r(x,x)\leq \int_{\mathbb{R}^d}dx\,\left[\int_\Omega\mu^{t,x}_{0,x}(d\omega)V_-(\omega(r)) \right]=
$$
$$
\int_{\mathbb{R}^d}dx\,\left[\int_\Omega\mu^{t,x}_{0,0}(d\omega)V_-(x+\omega(r))\right]=\overset{\circ}{\wp}_t(0)\int_{\mathbb{R}^d}dx\,V_-(x)\,<\,\infty,\quad\forall t>0.
$$
Thus, for any $r\in[0,t]$ the operator $S_r$ is trace class. Moreover, due to the properties of the trace we have $\mathsf{Tr}S_r=\mathsf{Tr}S_0$, $\forall r\in[0,t]$. We have:
$$
\mathsf{Tr}S_0=\frac{1}{t}\int_0^t dr\,(\mathsf{Tr}S_0)=\frac{1}{t}\int_0^t dr\,(\mathsf{Tr}S_r)=\frac{1}{t}\int_0^t dr\,\left[\int_{\mathbb{R}^d} dx\,K_r(x,x)\right]=
$$
$$
=\frac{1}{t}\int_{\mathbb{R}^d} dx\left[\int_\Omega\mu^{t,x}_{0,x}(d\omega)\tilde{g}_1\left(\int_0^t ds\,V_-(X_s)\right) \right]=\int_{\mathbb{R}^d} dx \Psi_1(t,x)
$$

In particular, for any $t>0$, $\Psi_{1}(t;\cdot)$ is integrable on $\Rd$.

On the other hand
$$
\rm{Tr}S_0=\int_{\Rd}K_0(x,x)dx=\int_{\Rd}dx\,V_-(x)\int_\Omega \mu^{t,x}_{0,x}(d\omega)
e^{-\int_0^t d\rho\,V_-(X_\rho)}
$$
$$
=\int_{\Rd}dx\,V_-(x)G_1(t;x,x)=\int_{\Rd}dx\, D_1(t;x).
$$

One uses the approximation properties contained in
Lemmas \ref{la} and \ref{le} as well as the Monotone Convergence Theorem.
\end{proof}

\begin{proof}{\it of Theorem \ref{Main}  for $B=0$}.

We can assume $V_+=0$ and $V_-\in C^\infty_0(\Rd)$. Lemma \ref{dolema} implies that
for any $\alpha>0$ one has
$$
N_{-\alpha}(-V_-)\le F_\infty(1)^{-1}\rm{Tr}\l[F_\infty(K_\alpha)\r].
$$
Using (\ref{once}), (\ref{doce}), we obtain
$$
\rm{Tr}\l[F_\infty(K_\alpha)\r]=\int_0^\infty dt\,e^{-\alpha t}\int_{\Rd} dx\,D_\infty(t;x)=
$$
\begin{equation}\label{assez}
=\int_0^\infty dt\,e^{-\alpha t}\int_{\Rd} dx\,\Psi_\infty(t;x).
\end{equation}
Inequality (\ref{main}) for $B=0$ follows from (\ref{assez}) and Lemma \ref{if}. In addition
$C_d=F_\infty(1)^{-1}\overline{C}_d$.
\end{proof}

\section{Proof of the bounds in the magnetic case.}

\begin{proof}{\it of Theorem \ref{Main} for} $B\ne 0$.

Analogously to Section 5, we can assume $V_+=0$ and $V_-\in C^\infty_0(\Rd)$. For $\alpha>0$ one sets
$K_\alpha(A):= V_-^{1/2}(H_A+\alpha)^{-1}V_-^{1/2}$. By inequality (\ref{unica}) for $r=1$ and also using
Pitt's Theorem \cite{P},
$K_\alpha(A)$ is a positive compact operator, and the same can be said about $F_\infty\l[K_\alpha(A)\r]$.
We show that $F_\infty\l[K_\alpha(A)\r]\in B_1$ and we estimate the trace-norm. As at the beginning of the proof
of Proposition \ref{l},
\begin{equation}\label{of}
F_\lambda\l[K_\alpha(A)\r]=V_-^{1/2}\int_0^\infty dt\,e^{-\alpha t}e^{-t(H_A+\lambda V_-)}V_-^{1/2}.
\end{equation}
By using Proposition \ref{FKI}, we get for any $u\in C_0(\Rd)$, $u\ge 0$
\begin{equation}\label{off}
\l[F_\lambda\l[K_\alpha(A)\r]u\r](x)=
\end{equation}
\begin{equation*}
=V_-^{1/2}(x)\int_0^\infty dt\,e^{-\alpha t}E_x\l[u(X_t)V_-^{1/2}(X_t)e^{-iS_A(t,X)}
g_\lambda\l(\int_0^t ds\,V_-(X_s)\r)\r].
\end{equation*}
Approximating $g_\infty$ by means of Lemmas \ref{la} and \ref{le} and using the Monotone Convergence Theorem, we see
that (\ref{off}) also holds for the pair $(g_\infty,F_\infty)$. The next inequality follows:
\begin{equation}\label{follows}
|F_\infty\l[K_\alpha(A)\r]u|\le F_\infty(K_\alpha)|u|,\ \ \forall u\in L^2(\Rd).
\end{equation}
By Lemma 15.11 from \cite{S1}, we have $F_\infty\l[K_\alpha(A)\r]\in B_1$ and
\begin{equation}\label{lab}
{\rm Tr}
\l(F_\infty\l[K_\alpha(A)\r]\r)\le {\rm Tr}\l(F_\infty\l[K_\alpha\r]\r).
\end{equation}
Denoting by $N_{-\alpha}(B,-V_-)$ the number of eigenvalues of $H_A-V_-$ strictly less than $-\alpha$,
analogously to Lemmas \ref{olema} and \ref{dolema}, we deduce that
\begin{equation}\label{shp}
N_{-\alpha}(B,-V_-)\le F_\infty(1)^{-1}{\rm Tr}\l(F_\infty\l[K_\alpha\r]\r).
\end{equation}
Inequality (\ref{main}) follows from (\ref{shp}) by using the estimations at the end of Section 5.
The constant $C_d$ is the same as for the case $B=0$.
\end{proof}

\begin{proof}{\it of Corollary \ref{LT}.}
The idea of the proof is standard (cf. \cite{S1} for instance), but one has to use parts of the arguments from the proof of Theorem \ref{Main} in the case $B=0$.

1. We show that it is enough to treat the case $V_+=0$.

We denote by $N$ (resp. $N_-$) the number of strictly negative eigenvalues of $H_A\dotplus V$
(resp. $H_A\dotplus(-V_-)$).
We have $N,N_-\in[0,\infty]$ and the min-max principle shows that $N\le N_-$. In addition, if $H_A\dotplus V$ has
strictly negative eigenvalues $\lambda_1\le \lambda_2\le \dots$, then $H_A\dotplus(-V_-)$ has strictly negative
eigenvalues $\lambda_1^-\le \lambda_2^-\le \dots$ and $\lambda_j^-\le \lambda_j$, $j\ge 1$. Therefore, one has
$\sum_{j\ge 1}|\lambda_j|^k\le\sum_{j\ge 1}|\lambda_j^-|^k$.

2. We show that treating compactly supported $V_-$ is enough (remark that this property implies that $V_-\in L^p(\mathbb{R}^d)$ for any $p\in[1,d+k]$).

We take into account the approximation sequence defined in Lemma \ref{cut-off-a}. The sequence of forms $(\mathfrak h^n)_{n\ge 1}$ satisfies
the hypothesis of Theorem 3.11, Ch. VIII from \cite{K}. If we denote by $\lambda_1\le \lambda_2\le \dots$
 the strictly negative eigenvalues of $H_A\overset{\cdot}{+}V$ and by $\lambda_1^{(n)}\le \lambda^{(n)}_2\le \dots$
 the strictly negative eigenvalues of $H^{(n)}:=H_A\overset{\cdot}{+}V^{(n)}$, once again by Theorem 3.15, Ch. VIII from \cite{K}, we have
 $\lambda_j^{(n)}\ge \lambda _j$, $\forall j,n\in\mathbb N^*$ and $\lambda_j^{(n)}$ converges to $\lambda_j$. So it will
 be sufficient to prove (\ref{main}) for the operators $H^{(n)}$.

3. We assume from now on that $V=-V_-$, $V_-\in L^{d+k}(\Rd)$ ($k>0$) and that  ${\rm supp}(V_-)$ is
compact. 
Let $\beta_0>0$ and for $\beta\in(0,\beta_0]$ let 
$$
\lambda_1\le\lambda_2\le\dots\le\lambda_{N_{-\beta}}<-\beta
$$ 
be the eigenvalues of $H=H_A\overset{\cdot}{+}(-V_-)$ strictly
smaller than $-\beta$ and let 
$$
\overline{\lambda}_1\le\overline\lambda_2\le\dots\le\overline\lambda_{M(\beta)}<-\beta
$$
be the distinct
eigenvalues with $m_j$ the multiplicity of $\overline{\lambda}_j$, $1\le j\le M(\beta)$. We have $N_{-\alpha}:=N_{-\alpha}(B,-V_-)$. Using the definition of the Stieltjes integral and integration by parts, we get
$$
\sum_{j=1}^{N_{-\beta}}|\lambda_j|^k=\sum_{j=1}^{M(\beta)} m_j|\overline\lambda_j|^k=\sum_{j=1}^{M(\beta)}|\overline\lambda_j|^k\l(
N_{\overline \lambda_{j+1}}-N_{\overline \lambda_{j}}\r)=\int_{\lambda_1}^{-\beta}|\lambda|^k dN_\lambda=
$$
\begin{equation}\label{final}
=|\beta|^k N_{-\beta}+k\int_{\lambda_1}^{-\beta}|\lambda|^{k-1}N_\lambda \,d\lambda.
\end{equation}
We denote by $I$ the last integral and use (\ref{shp}) and (\ref{assez}) and the arguments in the proof of Lemma \ref{if} to estimate $I$:
$$
I=\int_\beta^{-\lambda_1}\alpha^{k-1}N_{-\alpha}d\alpha=\left[F_\infty(1)\right]^{-1}\int_\beta^{-\lambda_1}\alpha^{k-1}\mathsf{Tr}F_\infty(K_\alpha)d\alpha=
$$
$$
=\left[F_\infty(1)\right]^{-1}\int_{\mathbb{R}^d}dx\int_0^\infty dt\,\Psi_\infty(t,x)\int_\beta^{-\lambda_1}d\alpha\,\alpha^{k-1}e^{-\alpha t}\leq
$$
$$
\leq \left[F_\infty(1)\right]^{-1}\int_{\mathbb{R}^d}dx\int_0^\infty dt\,t^{-1}\overset{\circ}{\wp}_t(0)\tilde{g}_\infty(tV_-(x))\int_\beta^{-\lambda_1}d\alpha\,\alpha^{k-1}e^{-\alpha t}\leq
$$
$$
\leq C\left[F_\infty(1)\right]^{-1}\int_{\mathbb{R}^d}dx\int_0^\infty dt\left(t^{-d-1}+t^{-d/2-1} \right)
\tilde{g}_\infty(tV_-(x))\int_\beta^{-\lambda_1}d\alpha\,\alpha^{k-1}e^{-\alpha t}
$$
The $\alpha$ integral may be bounded by:
$$
\int_0^{\infty}d\alpha\,\alpha^{k-1}e^{-\alpha t}=t^{-k}\int_0^{\infty}ds\,s^{k-1}e^{-s}\leq Ct^{-k}.
$$
Recalling that $\tilde{g}_\infty(t)=0$ for $t\leq1$ and $\tilde{g}_\infty(t)=t-1$ for $t>1$, we get that
$\tilde{g}_\infty(tV_-(x))=0$ for $V_-(x)=0$ and for $V_-(x)>0$
$$
\int_0^\infty dt\,t^{-k}\left(t^{-d-1}+t^{-d/2-1} \right)\tilde{g}_\infty(tV_-(x))=
$$
$$
=\left[V_-(x)\right]^{d+k}\int_1^\infty\,s^{-d-k-1}(s-1)ds\,+\,\left[V_-(x)\right]^{d/2+k}\int_1^\infty\,s^{-d/2-k-1}(s-1)ds,
$$
the integrals being convergent for $d\geq2$.

Using these estimations in (\ref{final}) we conclude that
$$
\sum_{j=1}^{N_{-\beta}}\left(|\lambda_j|^k-|\beta|^k\right)\leq C\left\{\int_{\mathbb{R}^d}\left[V_-(x)\right]^{d+k}dx\,+\,\int_{\mathbb{R}^d}\left[V_-(x)\right]^{d/2+k}dx\right\},
$$
thus
$$
\sum_{j=1}^{N_{-{(\beta_0)}}}\left(|\lambda_j|^k-|\beta|^k\right)\leq C\left\{\int_{\mathbb{R}^d}\left[V_-(x)\right]^{d+k}dx\,+\,\int_{\mathbb{R}^d}\left[V_-(x)\right]^{d/2+k}dx\right\},
$$
with the constant $C$ not depending on $\beta$ or $\beta_0$. We end the proof by leting $\beta\searrow0$.
\end{proof}

\subsubsection*{Acknowledgements}
VI and RP acknowledge partial support from the Contract no. 2-CEx06-11-18/2006.

\end{document}